\documentclass[conference]{IEEEtran}

\IEEEoverridecommandlockouts

\usepackage[a4paper, left=1.8cm, right=1.8cm, bottom=2.3cm, top=2.3cm]{geometry}

\usepackage{cite}
\usepackage{amsmath,amsthm,amssymb,amsfonts}
\usepackage{algorithmic}
\usepackage{graphicx}
\usepackage{textcomp}
\usepackage{xcolor}
\usepackage{bm}
\usepackage{cancel}
\usepackage{cleveref}

\newtheorem{definition}{Definition}
\newtheorem{theorem}{Theorem}
\newtheorem{lemma}{Lemma}
\newtheorem{proposition}{Proposition}
\newtheorem{corollary}{Corollary}
\newtheorem{claim}{Claim}
\newtheorem{remark}{Remark}

\renewcommand{\det}{D}

\newcommand{\hd}{H}
\newcommand{\Ber}{\mathrm{Ber}}

\newcommand\cH{\mathcal{H}}

\newcommand\cU{\mathcal{U}}
\newcommand\cX{\mathcal{X}}
\newcommand\cY{\mathcal{Y}}

\newcommand\bN{\mathbb{N}}
\newcommand\bP{\mathbb{P}}

\newcommand\indic{\mathbf{1}}

\newcommand\dtv{d_\mathrm{TV}}

\def\BibTeX{{\rm B\kern-.05em{\sc i\kern-.025em b}\kern-.08em
    T\kern-.1667em\lower.7ex\hbox{E}\kern-.125emX}}

\def \extended {1}

\begin{document}

\title{Secure Distributed Hypothesis Testing\thanks{The authors acknowledge the use of Gemini for editing and formatting, as well as for assisting in the development of some of the analytical arguments used to prove Proposition~\ref{prop:propFiniteRatio} and the reduction that follows. All such instances were critically examined, refined, and verified by the authors.
}
}

\author{
\IEEEauthorblockN{Gowtham R. Kurri}
\IEEEauthorblockA{\textit{IIIT Hyderabad}\\gowtham.kurri@iiit.ac.in }
\and
\IEEEauthorblockN{Varun Narayanan}
 \IEEEauthorblockA{\textit{CMI, Chennai}\\varunn@cmi.ac.in}
\and
\IEEEauthorblockN{Vinod M. Prabhakaran}
\IEEEauthorblockA{\textit{TIFR, Mumbai}\\vinodmp@tifr.res.in} 
\and
\IEEEauthorblockN{K. R. Sahasranand}
\IEEEauthorblockA{\textit{IIT Palakkad}\\sahasranand@iitpkd.ac.in}
}

\maketitle

\begin{abstract}
In distributed hypothesis testing, a central server performs hypothesis testing based on information received from distributed sensors/clients.
We study a secure variant of this problem in which the central server determines the hypothesis class of an underlying distribution without learning any additional information about the distribution itself. We prove that, in its standard form, this is impossible to achieve, even for simple and highly restricted cases. To bypass this impossibility, we augment the model with a shared secret key available to clients but hidden from the server. We show that a single-bit secret key enables perfectly secure testing for simple classes by reducing the test distributions to a symmetric, canonical instance. Finally, for arbitrary hypothesis classes over finite domains, we establish a reduction to standard hypothesis testing using Private Simultaneous Messages (PSM) protocols, achieving polynomial communication and key lengths.
\end{abstract}

\section{Introduction}
Consider the client-server model in which clients hold individual data and transmit locally computed messages to a central server tasked with evaluating a joint function of their combined data.
This computation model is highly practical and has been studied extensively in distributed computing.
In this setting, it is often critical to keep client data secure, in that, the server learns the function output but nothing else about client data.
This is the privacy notion considered in \emph{secure multiparty function computation} (MPC)\cite{Lin20}, a central primitive in cryptography and distributed computing that ensures that the server learns the function output and little else.

Distributed inference\cite{Tsitsiklis93} is a kind of computation in the client-server model where clients observe independent and identically distributed (i.i.d.) samples from an underlying distribution, and the server wants to learn a statistic of the distribution. Here, individual noisy samples (e.g., collected by distributed environmental sensors) forming the client data often lack sensitive meaning and do not require strict secrecy. Instead, a natural privacy goal is to protect the underlying distribution, ensuring that the server learns nothing about this distribution other than the target statistic.

We study this privacy notion in Distributed Hypothesis Testing (DHT)\cite{Tsitsiklis93} with composite hypotheses. Let the null hypothesis $\cH_0$ and the alternative $\cH_1$ represent distinct classes of distributions over a finite domain. In standard DHT, clients receive i.i.d. samples from an unknown test distribution $\mu \in \cH_0 \cup \cH_1$, and each client sends a local message to the server to aid detection of the hypothesis class, while optimizing the trade-off between total communication and correctness error.
We formulate Secure Distributed Hypothesis Testing (SDHT), which imposes an additional privacy constraint: the server must learn almost nothing about $\mu$ beyond its hypothesis class. 
Formally, the distributions of aggregated messages received by the server must be close in statistical distance for any pair of test distributions originating from the same class.
\subsection*{Our Contributions}
Our first contribution is definitional. SDHT devises the right model to study privacy of distributions during distributed inference. We conduct a systematic study of SDHT, detailing fundamental impossibilities and proposing workarounds.

We begin by showing that SDHT is impossible even for the simplest case where $\cH_0=\{\Ber(p_0),\Ber(p_1)\}$ vs. $\cH_1=\{\Ber(p_2)\}$ when $p_0,p_1,p_2$ are distinct; in that, both correctness and privacy error cannot be simultaneously made arbitrarily small even with arbitrarily many samples. To bypass this fundamental impossibility, we augment the model with a uniformly distributed shared secret key that is independent of the test distribution, available to all clients, and hidden from the server. 

Indeed, just a single bit of shared secret key suffices to bypass the impossibility for $\cH_0=\{\Ber(p), \Ber(1-p)\}$ and $\cH_1=\{\Ber(1/2)\}$, when $p\neq 1/2$. If the key is $0$, each client $i$ sends $Y_i=X_i$, where $X_i$ is the received sample; if the key is $1$, they send $Y_i=1-X_i$. Under both $\Ber(p)$ and $\Ber(1-p)$, the samples $Y_1,\ldots,Y_n$ are i.i.d. according to $\Ber(p)$ with probability $1/2$ and $\Ber(1-p)$ with probability $1/2$. Conversely, if the test distribution is $\Ber(1/2)$, the sequence remains i.i.d. $\Ber(1/2)$. This guarantees perfect privacy---$(Y_1,\ldots,Y_n)$ are identically distributed for both distributions in $\cH_0$---while still allowing the server to correctly identify the hypothesis class using a detector for $\cH_0$ vs. $\cH_1$.
We adapt this scheme to achieve perfectly secure hypothesis testing between any two distinct classes containing two distributions and one distribution, respectively, over an arbitrary finite domain.

Finally, we address SDHT for arbitrary hypothesis classes. We obtain a reduction from SDHT to standard (non-distributed) hypothesis testing, where the communication cost and secret key length scale polynomially with the number of samples (for a fixed domain size). This is achieved using Perfectly Private Simultaneous Messages (PSM) protocols: a primitive from information-theoretically secure multiparty computation. 

\subsection*{Related Work}
The study of distributed hypothesis testing dates back to the works of Tenney and Sandell~\cite{TenSan81} and of Ahlswede and Csiszar~\cite{AhlCsi86}.
The setting where each client/sensor gets a single sample and is rate limited was studied in \cite{Tsitsiklis1988}.
DHT has been well-studied\cite{HanAmari98, Varshney96, RahWag10}, with a focus on various aspects such as error exponents~\cite{escamilla2020distributed, hadar2019error}, sample complexity~\cite{kazemi2025sample}, communication complexity~\cite{andoni2019two, sahasranand2021communication}, and shared randomness~\cite{acharya2019communication}.

Various notions of security are considered for computation in the client server model.
Secure multiparty computation~\cite{Lin20}---a central notion in cryptography---adopts an indistinguishibility-based notion of security where the messages received by the server are required to be close in statistical distance for any pair of inputs that evaluate to the same value for the pre-agreed function.
Most relevant to our work is a primitive in MPC, namely, private simultaneous messages (PSM) protocols \cite{FKN94} where clients share a secret key unknown to the server who wants to compute a function of the client data while maintaining the aforementioned security.

A different line of work, in statistics, computer science, and other related fields, concerns differential privacy, introduced in \cite{Dwork06}, and extensively studied thereafter: see the survey by Dwork \cite{Dwork08}. Ensuring differential privacy in the client-server model leads to the study of local differential privacy~\cite{KLNRS11,DJW13,acharya2018inspectre}. Various other notions of privacy and specialized privacy-preserving tasks have been explored, including Pufferfish privacy \cite{nuradha2023pufferfish}, maximal leakage \cite{issa2019operational,LiaoSCT17}, privacy against hypothesis-testing adversaries \cite{LiOG19}, and private mean estimation \cite{agarwal2025private}. These notions are primarily about protecting client-level data, whereas our problem concerns the privacy of the statistic of the distribution. Thus, the setting we consider is fundamentally different from the above line of work.

\subsection*{Notation}
Finite sets are denoted as $\cX,\cY,\mathcal{K}$ and so on; a distribution over $\cX$ by $P_X$, etc;
random variables distributed according a $P_X$ by $X$.
When there is no room for ambiguity, a sequence $(X_1,X_2,\ldots,X_n)$ will be shortened as $X^n$.
We use $W_{Y|X} \circ P_X$ to denote the distribution of $Y$ with probability mass function $\sum_{x \in \cX} P_X(x)W_{Y|X}(\cdot|x)$. Total variation distance/statistical distance between two distributions $P$ and $Q$ over a finite domain $\cX$ is denoted $\dtv(P,Q)$ and given by $\dtv(P,Q) = \frac{1}{2}\sum_{x\in \cX}|P(x) - Q(x)|$.
The Bernoulli distribution with mean $p$ is denoted $\Ber(p)$. WLOG, as usual, stands for \emph{without loss of generality}.

\section{Problem Statement}
\begin{definition}[Binary Hypothesis Testing]
    Let $\cH_0$ and $\cH_1$ be two classes of distributions over a finite domain $\cX$.
    For $n\in\bN$, $\varepsilon\ge 0$, a $(1-\varepsilon)$-correct hypothesis test for $\mathcal{H}_0$ vs. $\mathcal{H}_1$ using $n$ samples is a deterministic function $\det:\cX^n \to \{0,1\}$ such that, when $X_1,\ldots,X_n \sim \mu$, i.i.d.,
    \[
        \bP \left[\det\left(X_1,\ldots,X_n\right)=b\right] \geq 1-\varepsilon,\; \forall \mu\in\cH_b,\, b\in\{0,1\}.
    \]
\end{definition}
\begin{definition}\label{def:private_test}
    Let $\cH_0$ and $\cH_1$ be two classes of distributions over a finite domain $\cX$.
   
    For $n\in\bN$, $\varepsilon,\delta\ge 0$, an $(\varepsilon,\delta)$-secure distributed hypothesis testing (SDHT) scheme for $\cH_0$ vs. $\cH_1$ using $n$ samples is defined by randomized functions $W_{Y_i|X_i}$ with output alphabet $\cY_i$ for each $i\in[n]$ and a deterministic function $\det:\prod_{i=1}^n \cY_i \to \{0,1\}$. Define $P^\mu_{Y^n} = \prod_{i=1}^n \left(W_{Y_i|X_i}\circ \mu\right)$ for any $\mu\in\cH_0\cup \cH_1$. Then, the following conditions hold:

    \paragraph{$(1-\varepsilon)$-correctness}
    When $(Y_1,\ldots,Y_n) \sim P_{Y^n}$,
    \[
        \bP \left[\det\left(Y_1,\ldots,Y_n\right)=b\right] \geq 1-\varepsilon, \forall \mu\in\cH_b,\, b\in\{0,1\}.
    \]
    \paragraph{$\delta$-privacy} For all $b\in\{0,1\}$ and $\mu,\mu'\in\cH_b$,
    \begin{align*}
        \dtv\left(P^\mu_{Y^n},P^{\mu'}_{Y^n}\right) \leq \delta.
    \end{align*}
    The SDHT scheme uses $\sum_{i=1}^n \log |\cY_i|$ bits of communication.
\end{definition}
\begin{remark}\label{remark} We will use the following simple observation in our proofs. By standard results (see, for instance,~\cite{tyagi2023information}; for completeness, we also include a proof at the end of the appendix), $\dtv(P,Q) \ge \delta$ if and only if there exists a deterministic test $D$ that uses sample $X$ and outputs $D(X) \in \{0,1\}$ such that
\[
\bP_P[D(X) = 0] + \bP_Q[D(X) = 1] \ge 1+\delta.
\]
\end{remark}
We will also consider SDHT using \emph{shared secret key} where each client $i$ computes $Y_i$ conditioned on $X_i$ and a sampled shared secret key $K$ distributed uniformly over a finite set $\mathcal{K}$.
An $(\varepsilon,\delta)$-SDHT scheme for $\cH_0$ vs. $\cH_1$ using shared secret key over a finite domain $\mathcal{K}$ and $n$ samples is defined by randomized functions $W_{Y_i|X_i,K}$ with output alphabet $\cY_i$ and a deterministic function $\det:\prod_{i=1}^n\cY_i \to \{0,1\}$.
When $P_K$ is the uniform distribution over $\mathcal{K}$ and $\mu\in\cH_0\cup\cH_1$, define
$$P^\mu_{X^n K Y^n} = \mu^{\otimes n} \cdot P_K \cdot \prod_{i=1}^n W_{Y_i|X_i,K}.$$
The correctness and privacy properties in \Cref{def:private_test} hold with respect to $P^\mu_{Y^n}$ (the marginal distribution of $Y^n$ under $P^\mu_{X^n K Y^n}$) for all $\mu\in\cH_0\cup \cH_1$. The communication incurred is $\sum_{i=1}^n|\cY_i|$ bits. The secret key size is $\log|\mathcal{K}|$ bits.

Hypotheses $\cH_0$ and $\cH_1$ are \emph{securely distinguishable} without shared key (respectively, with shared key) if there exists an $(\varepsilon,\delta)$-SDHT scheme using shared key (respectively, with shared key) for $\cH_0$ vs. $\cH_1$ for any $\varepsilon,\delta>0$ using $n$ samples when $n$ is sufficiently large.
\section{Main Results}
Suppose there is a randomized map (or channel) that maps every distribution in $\cH_0$ to the same message distribution and every distribution in $\cH_1$ to the same message distribution (different from the message distribution under $\cH_0$).
If each client computes their message by applying this map to their sample, \emph{perfect} privacy is guaranteed: indeed, messages received by the server are identically distributed for all distributions in $\cH_0$; similarly for $\cH_1$.
To detect the hypothesis class, the server can perform a simple hypothesis test between the message distribution under $\cH_0$ vs. that under $\cH_1$.
This is formalized in the following proposition:
\begin{proposition}\label{prop:lp}
   
    Suppose there exists $W_{Y|X}$ and distinct distributions $P^{(0)}_Y$ and $P^{(1)}_Y$ such that $W_{Y|X}\circ\mu = P^{(b)}_Y$ for every $\mu\in\cH_b$ and $b\in\{0,1\}$. Then, there exists an
    $(e^{-\Omega(n)},0)$-SDHT scheme
    for $\cH_0$ vs. $\cH_1$ using $n$ samples that incurs $n\cdot \log|\cY|$ bits of communication and does not use a shared secret key.
\end{proposition}
Unfortunately, such a channel does not exist in general; specifically, even for the simple case of $\{\Ber(p_0),\Ber(p_1)\}$ vs. $\{\Ber(p_2)\}$ for distinct $p_0,p_1,p_2$, such a channel does not exist.
Our main technical contribution is showing that, for this pair of hypotheses, it is impossible to drive both correctness and privacy error arbitrarily small without using a shared secret key no matter how many samples are used.

\begin{theorem}
    
    For any distinct $p_0,p_1,p_2\in[0,1]$, the hypotheses $\{\Ber(p_0),\Ber(p_1)\}$ and $\{\Ber(p_2)\}$ are not securely distinguishable without using shared secret keys.
    \label{thm:impossibility}
\end{theorem}
Using \Cref{thm:impossibility} and \Cref{prop:lp}, we can obtain the following general result.
\begin{corollary}\label{corr:collinear}
   
    Let $\mu_0,\mu_1,\mu_2$ be distributions over the same domain.
    Then, hypotheses $\{\mu_0,\mu_1\}$ and $\{\mu_2\}$ are securely distinguishable without shared secret key if and only if there exists no $\theta \in[0,1]$ and distinct $a,b,c\in\{0,1,2\}$ such that $\theta\mu_a+(1-\theta)\mu_b=\mu_c$. 
\end{corollary}

We bypass the above impossibility by augmenting the model with a uniformly distributed shared secret key.
This key is independent of the test distribution, available to all clients, and hidden from the server.
\begin{theorem}\label{thm:achievability}
    Let $\mu_0,\mu_1,\mu_2$ be distributions over the same domain.
    There exists a $(e^{-\Omega(n)},0)$-SDHT scheme for $\{\mu_0,\mu_1\}$ vs. $\{\mu_2\}$ using $n$ samples and a $1$-bit shared secret key.
\end{theorem}

Finally, we address SDHT for arbitrary hypothesis classes in the shared-key setting.
We demonstrate a reduction from SDHT to standard (non-distributed) hypothesis testing, where the communication cost and secret key length scale polynomially with the number of samples (for a fixed domain size).

This is achieved using Perfectly Private Simultaneous Messages (PSM) protocols: a primitive from information-theoretically secure multiparty computation. 
These protocols allow the server to compute a predetermined function of the clients' inputs while ensuring perfect privacy against the server. 
Let $\det:\cX^n \to \{0,1\}$ be a detector for the standard hypothesis test between $\cH_0$ and $\cH_1$ using $n$ samples. 
By applying a PSM protocol to $\det$, we obtain a distributed test with correctness and privacy errors upper-bounded by the correctness error of $\det$. 
Because $\det$ is a symmetric function, we utilize Eriguchi and Shinagawa PSM protocol~\cite{ES25} for symmetric functions, yielding exponential communication efficiency compared to general PSM protocols.

\begin{theorem}\label{thm:psm}
 
    Let $\det:\prod_{i=1}^n\cX^n$ be a symmetric function achieving $(1-\varepsilon)$-correct hypothesis test for $\mathcal{H}_0$ vs. $\mathcal{H}_1$ using $n$ samples.
    There exists an $(\varepsilon,\varepsilon)$-SDHT scheme for $\cH_0$ vs. $\cH_1$ using $n$ samples, incurring $O(n^{2\lceil |\cX|/3\rceil})$ bits of communication and shared keys.
\end{theorem}

\section{Proofs}
\subsection*{Proof of Theorem~\ref{thm:impossibility}}
Any details omitted from this proof are provided in
\if \extended 1%
the Appendix.
\fi
\if \extended 0%
\cite[Appendix]{sdht2026}.
\fi
It suffices to consider the cases $p_0<p_2<p_1$ and $p_0<p_1<p_2$; the other cases are similarly handled by relabelling $p_i$.
    The former is relatively straightforward to establish, while the latter requires a more delicate argument.

    (a) \emph{Case $p_0<p_2<p_1$}: Let $p_2=\theta\cdot p_1 + (1-\theta)\cdot p_0$ for some $\theta \in (0,1)$.
    Towards a contradiction, suppose there exists an $(\varepsilon,\delta)$-SDHT scheme $\left(\{W_{Y_i|X_i}\}_{i=1}^n,\det\right)$ for this problem using $n$ samples and without using a shared key.
    Let,
    \[ (Y^{(b)}_1,\ldots,Y^{(b)}_n) \sim \prod_{i=1}^n W_{Y_i|X_i}\circ \Ber(p_b), \forall b\in\{0,1,2\}.\]
    By ($1-\varepsilon$)-correctness, $\det(Y^{(b)}_1,\ldots,Y^{(b)}_n)=c$ with probability at least $1-\varepsilon$ for each $(b,c)\in\{(0,0),(1,0),(2,1)\}$.
    We will construct a randomized map $\det'$ such that $\det'(Y^{(b)}_1,\ldots,Y^{(b)}_n)=b$ with probability at least $1-\varepsilon$ for $b\in\{0,1\}$.
    This contradicts $\delta$-privacy for sufficiently small $\varepsilon>0$ and $\delta>0$.

    $\det'$ is defined as follows: On input $Y_1,\ldots,Y_n$, for each $i\in[n]$, sample $Y'_1,\ldots,Y'_n$ i.i.d. according to $W_{Y_i|X_i}\circ \Ber(p_0)$ and $B_1,\ldots,B_n$ i.i.d. according to $\Ber(\theta)$.
    For each $i\in[n]$, let $Z_i = B_i\cdot Y_i + (1-B_i)\cdot Y'_i$, and output $\det(Z_1,\ldots,Z_n)$.
    If each $Y_i$ is distributed i.i.d. according to $W_{Y_i|X_i}\circ \Ber(p_0)$, then each $Z_i$ is also distributed i.i.d. according to $W_{Y_i|X_i}\circ \Ber(p_0)$.
    But, if each $Y_i$ is distributed i.i.d. according to $W_{Y_i|X_i}\circ \Ber(p_1)$, then each $Z_i$ is distributed i.i.d. according to $W_{Y_i|X_i}\circ \Ber(p_2)$.
    This follows from the fact that, for any $W_{Y|X}$, $$W_{Y|X}\circ \Ber(p_2) = W_{Y|X}\circ \left(\theta\cdot \Ber(p_1) + (1-\theta)\cdot \Ber(p_0)\right).$$
    The properties of $\det'$ follow from the correctness of $\det$.

    (b) \emph{Case $p_0<p_1<p_2$}:
    There exists $\theta\in[0,1]$,
    such that $p_1=\theta\cdot p_2 + (1-\theta)\cdot p_0$.
    Consider $R_{X|U}$ where $R_{X|U}(\cdot|0)=\Ber(p_0)$ and $R_{X|U}(\cdot|1)=\Ber(p_2)$.
    Then, $R_{X|U}\circ \Ber(0) = \Ber(p_0)$, $R_{X|U}\circ \Ber(1) = \Ber(p_2)$ and $R_{X|U}\circ \Ber(\theta) = \Ber(p_1)$.
    Thus, if $(\{W_{Y_i|X_i}\}_{i=1}^n, \det)$ is an $(\varepsilon,\delta)$-SDHT scheme for $\{\Ber(p_0),\Ber(p_1)\}$ vs. $\{\Ber(p_2)\}$, then $(\{W_{Y_i|X_i}\circ R_{X|U}\}_{i=1}^n, \det)$ is an $(\varepsilon,\delta)$-SDHT scheme for $\{\Ber(0),\Ber(\theta)\}$ vs. $\{\Ber(1)\}$.
Hence, it suffices to show that theorem holds for $p_0=0, p_1=\theta$ and $p_2=1$ for all $\theta\in (0,1)$.

It suffices to restrict to SDHT schemes $(\{W_{Y_i|X_i}\}_{i=1}^n,\det)$ where $W_{Y_i|X_i}(\cdot|0) \neq W_{Y_i|X_i}(\cdot|1)$ for all $i\in[n]$.
Clearly, $Y_j$ is independent of $\{Y_i\}_{i\neq j}$ and identically distributed irrespective of the distribution of $X_j$.
Hence, for all $a,b\in\{0,1,2\}$, denoting $(Y_1,\ldots,Y_{j-1},Y_{j+1},\ldots,Y_n)$ as $Y^{-j}$,
\begin{align*}
\dtv\left(P^{(a)}_{Y^n},P^{(b)}_{Y^n}\right)
= \dtv\left(P^{(a)}_{Y^{-j}},P^{(b)}_{Y^{-j}}\right).
\end{align*}
In other words, $(\{W_{Y_i|X_i}\}_{i=1}^n,\det)$ is an $(\varepsilon,\delta)$-SDHT scheme (if and) only if $(\{W_{Y_i|X_i}\}_{i\neq j},\det)$ is an $(\varepsilon,\delta)$-SDHT scheme.
\begin{proposition}
    Suppose $P^{(j)}_{Y_i} \equiv W_{Y_i|X_i} \circ ~\Ber(p_j)$ for each $j \in \{0,1,2\}$ and $i\in[n]$. Suppose
    \begin{align}\label{prop5eq0}
        \inf_{\substack{W_{Y_i|X_i}:\\ W_{Y_i|X_i}(\cdot|0) \neq W_{Y_i|X_i}(\cdot|1)}} \frac{\hd^2\left(P^{(0)}_{Y_i},P^{(1)}_{Y_i}\right)}{\hd^2\left(P^{(1)}_{Y_i},P^{(2)}_{Y_i}\right)} = \lambda > 0.
    \end{align}
    Then, either
    $
        \dtv(P^{(1)}_{Y^n},\,P^{(2)}_{Y^n}) \le 1/2
    $
    or
    $
        \dtv(P^{(0)}_{Y^n},\,P^{(1)}_{Y^n}) \ge
        1-e^{-\left(\sqrt{3}/2-1\right)\lambda}
    $ for all $n$ and all schemes using $n$ samples.
    
\label{prop:privCorrTradeoff}
\end{proposition}
\begin{proof}
    Fix $n$ and a SDHT scheme $\{\{W_{Y_i|X_i}\}_{i=1}^n,\det\}$.
    Let $P^{(j)}_{Y^n} = \prod_{i=1}^n W_{Y_i|X_i}\circ \Ber(p_j)$ for $j\in\{0,1,2\}$.
    Suppose $\dtv\left(P^{(1)}_{Y^n},\,P^{(2)}_{Y^n}\right) \ge 1/2$; else we are done.
    From the relationship between total variation distance and squared Hellinger distance~\cite[Equation (7.22)]{polyanskiy_wu_itbook},
    \begin{align}
        \dtv\left(P,Q\right)\leq \sqrt{1-\left(1-\frac{1}{2}\hd^2(P,Q)\right)^2}.
    \end{align}
    Hence,
    \begin{align}
        H^2\left(P^{(1)}_{Y^n},\,P^{(2)}_{Y^n}\right)   & \geq 2-2\sqrt{1-\dtv\left(P^{(1)}_{Y^n},\,P^{(2)}_{Y^n}\right)^2}\nonumber \\
        & \geq 2(1-\sqrt{(1-(1/2)^2)})
         = 2-\sqrt{3} \label{prop5eq01}.
        \end{align}
    Next,
    \begin{align}
        2\dtv   \left(P^{(0)}_{Y^n},\,P^{(1)}_{Y^n}\right)\ge{} &  \hd^2\left(P^{(0)}_{Y^n},\,P^{(1)}_{Y^n}\right)\label{prop5eq1}                                                     \\
        ={}&     2 - 2 \prod_{i=1}^n \left(1-\frac{1}{2}H^2\left(P^{(0)}_{Y_i},P^{(1)}_{Y_i}\right)\right)\label{prop5eq2}         \\
        \ge{}&   2 - 2\mathsf{e}^{-\frac{1}{2}\sum_{i=1}^nH^2\left(P^{(0)}_{Y_i},P^{(1)}_{Y_i}\right)} \label{prop5eq3}        \\
        \ge{} &  2 - 2\mathsf{e}^{-\frac{\lambda}{2}\sum_{i=1}^nH^2\left(P^{(1)}_{Y_i},\,P^{(2)}_{Y_i}\right)}\label{prop5eq4} \\
        \geq{} & 2 - 2\mathsf{e}^{-\frac{\lambda}{2}H^2\left(P^{(1)}_{Y^n},\,P^{(2)}_{Y^n}\right)}\label{prop5eq5}             \\
        \geq{} & 2\left(1-\mathsf{e}^{-\left(\sqrt{3}/2-1\right)\lambda}\right),\label{prop5eq6}
    \end{align}
    where \eqref{prop5eq1} follows from~\cite[Equation (7.22)]{polyanskiy_wu_itbook}, \eqref{prop5eq2} follows from the tensorization property for squared Hellinger distance under product distributions~\cite[Equation (7.26)]{polyanskiy_wu_itbook}, \eqref{prop5eq3} follows from $1-x\leq \mathsf{e}^{-x}$, for all $x\in\mathbb{R}$, \eqref{prop5eq4} follows from the definition of $\lambda$ in \eqref{prop5eq0}, \eqref{prop5eq5} follows from the tensorization property together with the identity $1-\prod_{i=1}^n(1-x_i)\leq \sum_{i=1}^nx_i$, for all $x_i\in[0,1]$, \eqref{prop5eq6} follows from the lower bound in \eqref{prop5eq01}. 
\end{proof}
Next, we show that the assumption in Proposition~\ref{prop:privCorrTradeoff} is satisfied. WLOG, suppose that $\cH_0 = \{\Ber(0), \Ber(\theta)\}$ for a fixed $\theta \in (0, 1)$ and that $\cH_1= \{\Ber(1)\}$. 
\begin{proposition}
    We have
    \[
        \sup_{W_{Y_i|X_i}: W_{Y_i|X_i}(1|0) \neq W_{Y_i|X_i}(1|1)} \frac{\hd^2\left(P^{(1)}_{Y_i},P^{(2)}_{Y_i}\right)}{\hd^2\left(P^{(0)}_{Y_i},P^{(1)}_{Y_i}\right)} < \infty.
    \]
\label{prop:propFiniteRatio}
\end{proposition}

\emph{Proof sketch.} For a channel $W_{Y|X}$ with $W_{Y|X}(1|0) = q_0$ and $W_{Y|X}(1|1) = q_1,~q_0 > q_1$, WLOG, we can write
    \begin{align*}
        P_{Y}^{(2)} & \equiv \Ber(a),\,P_{Y}^{(1)} \equiv \Ber(a + c'),\,
        P_{Y}^{(0)}  \equiv \Ber(a+c),
    \end{align*}
    where $a = q_1,\, c = q_0-q_1,$ and $c' = (1-\theta)c$. Define
    \begin{align*}
        f(a,c) & := \frac{\hd^2\left(P^{(1)}_{Y},P^{(2)}_{Y}\right)}{\hd^2\left(P^{(0)}_{Y},P^{(1)}_{Y}\right)}                                        \\
               = &\frac{1-\sqrt{a(a+c')}-\sqrt{(1-a)(1-a-c')}}{1-\sqrt{(a+c')(a+c)}-\sqrt{(1-a-c')(1-a-c)}}.
    \end{align*}
     The proof is essentially analytical; it entails showing that $f(a,c)$ is non-increasing in $a$ for a fixed $c$, and showing that $f(0,c)$ is decreasing in $c$ whereby 
    
     \[
     f(a,c) \le f(0,c) \le \lim_{c \downarrow 0} f(0,c) < \infty. 
     \]
\hfill $\square$\\
Finally, we show that for channels with finite output alphabet ($|\cY| < \infty$), the supremum of squared Hellinger distance ratios of the output distributions does not exceed that of the output distributions generated by channels with binary output alphabet.

For $\cY = [k]$, let $W_{Y|X}$ be a channel such that
\[ 
W_{Y|X}(i|0)=(\alpha^{(i)}_0)_{i=1}^k \text{ and } W_{Y|X}(i|1)=(\alpha^{(i)}_1)_{i=1}^k, 
\]
where 
\[
\frac{\alpha^{(i)}_0}{\alpha^{(i)}_1} \le \frac{\alpha^{(j)}_0}{\alpha^{(j)}_1} \; \forall i\le j.
\]
We do this in three steps:
\begin{itemize}
    \item[i.] We consider the transformation $W \rightarrow W^{\gamma}$ for $\gamma\in(0,1)$ such that, for $b\in\{0,1\}$,
\begin{align*}
    P_{W^{\gamma}}(Y=i|b)=
    \begin{cases}
        (1-\gamma)\cdot \alpha^{(1)}_b + \gamma & \text{for} \ i=1         \\
        (1-\gamma)\cdot \alpha^{(i)}_b          & \text{for} \ i=2,\ldots,k
    \end{cases}
\end{align*}
and prove that there exists $\gamma^*\in(0,1)$ such that 
    \begin{align*}
        \frac{(1-\gamma)\cdot \alpha_0^{(1)}+\gamma}{(1-\gamma)\cdot \alpha_1^{(1)}+\gamma}\le \frac{\alpha_0^{(2)}}{\alpha_1^{(2)}}~;~ \gamma\in[0,\gamma^*]
    \end{align*}
    with equality at $\gamma = \gamma^*$.
    \item[ii.] Next, we show that the squared Hellinger distance ratio for the output distributions generated by $W^\gamma$ is upper bounded by that of the output distribution generated by $W$.
    \item[iii.] Finally, we show that the squared Hellinger ratio of the output distributions does not change if we ``merge'' two symbols with the same $\alpha_0/\alpha_1$ ratio.
\end{itemize}
By repeating this process in {\em{(i)-(iii)}}, we can reduce the alphabet size to obtain a  binary channel.
    \hfill $\square$

\subsection*{Proof sketch of Corollary~\ref{corr:collinear}}
   A detailed proof is provided in 
\if \extended 1%
Appendix.
\fi
\if \extended 0%
\cite[Appendix]{sdht2026}.
\fi The ``if" direction is shown by presenting a channel $W_{Y|X}$ that maps $\mu_0,\mu_1$ to the same distribution and $\mu_2$ to a distinct distribution whenever such $\lambda,a,b,c$ do not exist, and then appealing to \Cref{prop:lp}.

    To see the ``only if" direction, suppose $\theta\mu_2 + (1-\theta)\mu_0 = \mu_1$ for some $\theta \in [0,1]$. 
    By employing the same reduction considered in the beginning of case where $p_0<p_1<p_2$ in the proof of \Cref{thm:impossibility},
    we can argue that $\cH_0$ and $\cH_1$ are securely distinguishable only if so are hypotheses $\{\Ber(0),\Ber(\theta)\}$ and $\{\Ber(1)\}$.
    The result now follows from \Cref{thm:impossibility}.
    
    A similar argument can be used when $\theta,a,b,c$ exists as stated in the corollary for any distinct $a,b,c$.\hfill$\square$

\subsection*{Proof Sketch of Theorem~\ref{thm:achievability}}

    A detailed proof is provided in 
    \if \extended 1%
Appendix.
\fi
\if \extended 0%
\cite[Appendix]{sdht2026}.
\fi As discussed in the introduction, such a scheme can be obtained for $\{\Ber(p_0),\Ber(p_1)\}$ vs $\{\Ber(p_2)\}$, using a single bit of shared key if $p_0=1-p_1$ and $p_2=1/2$.
    This can be straightforwardly adapted for the case where $p_1=1-p_0$ and $p_2\notin\{p_0,p_1\}$.
    We build $W_{Y|X}$ that maps $\Ber(p_0),\Ber(p_1)$ and $\Ber(p_2)$ to $\Ber(p),\Ber(1-p)$ and $\Ber(q)$, respectively, such that $q\notin\{p,1-p\}$ whenever $p_0,p_1,p_2$ are distinct.
    This reduces testing $\{\Ber(p_0),\Ber(p_1)\}$ vs. $\{\Ber(p_2)\}$ to testing $\{\Ber(p),\Ber(1-p)\}$ vs. $\{\Ber(q)\}$.

    For any finite domain $\cX$, we build a deterministic map $W_{Y|X}:\cX\to\{0,1\}$ such that $W_{Y|X}\circ \mu_0,\;W_{Y|X}\circ \mu_1$ and $W_{Y|X}\circ \mu_2$ are distinct Bernoulli distributions whenever $\mu_2,\mu_1,\mu_2$ are distinct distributions over $\cX$.

\subsection*{Proof of Theorem~\ref{thm:psm}}
We will require the following definitions in the proof.
\begin{definition}
    A Boolean function $f:\cX^n\to\{0,1\}$ is symmetric if for any $(x_1,\ldots,x_n)\in\cX^n$ and any permutation $\sigma$ on $[n]$, $f(x_1,\ldots,x_n) = f(x_{\sigma(1)},\ldots,x_{\sigma(n)})$.
\end{definition}
\begin{definition}
    Let $f:\cX^n\to\{0,1\}$ be an $n$-variate Boolean function. A private simultaneous messages (PSM) protocol for $f$ using secret key $K$ uniformly distributed over a domain $\mathcal{K}$ is a tuple $\left(\{W_{Y_i|X_i,K}\}_{i=1}^n,\phi\right)$ where each $W_{Y_i|X_i,K}$ is a channel with output alphabet $\cY_i$, and $\phi$ is a deterministic function $\det:\prod_{i=1}^n \cY_i \to \{0,1\}$ satisfying the following properties with respect to
    \[ P_{X^nKY^n} = \prod_{i=1}^n P_X \cdot P_K \cdot \prod_{i=1}^n W_{Y_i|X_i,K}. \]
    where $P_K$ and $P_X$ are uniform distributions over $\mathcal{K}$ and $\cX$, respectively.

    \paragraph{Correctness} For all $(x_1,\ldots,x_n)\in\cX^n$ and $k\in\mathcal{K}$, 
    \[ \Pr[\det(Y_1,\ldots,Y_n) = f(x_1,\ldots,x_n) | X^n=x^n, K=k] = 1.\]

    \paragraph{Privacy} For any $(x_1,\ldots,x_n),(\tilde{x}_1,\ldots,\tilde{x}_n)\in\cX^n$ such that $f(x_1,\ldots,x_n) = f(\tilde{x}_1,\ldots,\tilde{x}_n)$, 
    \[ P_{Y^n|X^n=x^n} = P_{Y^n|X^n=\tilde{x}^n}.\]

    The communication incurred is $\sum_{i=1}^n|\cY_i|$ bits. The secret key size is $\log|\mathcal{K}|$ bits.
\end{definition}

We will use the following result from \cite{ES25} on the existence of PSM protocols for symmetric functions.
\begin{theorem}[Theorem 1 in \cite{ES25}]\label{thm:ES25}
    For any symmetric function $f:\cX^n\to\{0,1\}$, there exists a PSM protocol for $f$ using that incurs $O(n^{\delta})$ bits of communication and $O(n^\delta)$ key length, where $\delta=2\lceil |\cX|/3 \rceil + 1$.
\end{theorem}

    There exists a PSM $\left(\{W_{Y_i|X_i,K}\}_{i=1}^n,\phi\right)$ for $\det$ with $O(n^{\delta})$ bits of communication and $O(n^\delta)$ key length, where $\delta=2\lceil |\cX|/3 \rceil + 1$. This follows from \Cref{thm:ES25} since $\det$ is symmetric.
    We will argue that $\left(\{W_{Y_i|X_i,K}\}_{i=1}^n,\phi\right)$ is an $(\varepsilon,\varepsilon)$-private distributed hypothesis test for $\cH_0$ vs. $\cH_1$ using $n$ samples with the same communication and key length.

    For any $b\in\{0,1\}$ and $\mu\in\cH_b$, when $X_1,\ldots,X_n$ are i.i.d. according to $\mu$, by the correctness of the PSM protocol, the output of the test is $\det(X_1,\ldots,X_n)$. But then, since $\det$ is a $(1-\varepsilon)$-correct hypothesis test for $\cH_0$ vs. $\cH_1$, the output of the test is $b$ with probability at least $1-\varepsilon$. This establishes the correctness of the distributed hypothesis test.

    Next, when $X_1,\ldots,X_n$ are i.i.d. according to $\mu$, define 
    \begin{align}\label{eq:psm-distribution}
        P_{X^nKY^n} = \mu^{\otimes n} \cdot P_K \cdot \prod_{i=1}^n W_{Y_i|X_i,K}.
    \end{align}
    For any $\mu$, when $X_1,\ldots,X_n$ are i.i.d. according to $\mu$, suppose $\det(X_1,\ldots,X_n)=0$ with probability $\alpha$. Let $P_{Y^n}$ and $P_{Y^n|X^n}$ be defined with respect to the distribution $P_{X^nKY^n}$ in \eqref{eq:psm-distribution}. Fix any $(x_1,\ldots,x_n)\in\cX^n$ such that $\det(x_1,\ldots,x_n)=0$ and any $(\tilde{x}_1,\ldots,\tilde{x}_n)\in\cX^n$ such that $\det(\tilde{x}_1,\ldots,\tilde{x}_n)=0$. Then,
    \[ P_{Y^n} = \begin{cases}
        P_{Y^n|X^n=x^n} & \text{ with probability } \alpha, \\
        P_{Y^n|X^n=\tilde{x}^n} & \text{ with probability } 1-\alpha.
    \end{cases}\]
    Since $\det$ is a $(1-\varepsilon)$-correct hypothesis test for $\cH_0$ vs. $\cH_1$, for any $\mu\in\cH_0$, $\alpha\geq 1-\varepsilon$. Hence,
    $$\dtv(P_{Y^n},P_{Y^n|X^n=x^n}) \leq \varepsilon.$$
    Similarly, for any $\mu\in\cH_1$, $\alpha\leq \varepsilon$. Hence,
    $$\dtv(P_{Y^n},P_{Y^n|X^n=\tilde{x}^n}) \leq \varepsilon.$$ This establishes $\varepsilon$-privacy of the distributed hypothesis test. The communication and key length of the distributed hypothesis test are the same as that of the PSM protocol, which are $O(n^{\delta})$ bits and $O(n^\delta)$ bits, respectively. \hfill $\square$

\section{Discussion and Future Work}
For arbitrary hypothesis classes, (i) we need a shared secret key of substantial length, and (ii) it needs to be kept completely secret from the server.

When privacy is required only against a computationally bounded server, the former concern can be addressed by using a pseudorandom generator to stretch a short seed into a long pseudorandom string that can be used as the shared secret key. 
Addressing the latter concern is an interesting and important direction for future work. 
The impossibility result in Theorem~\ref{thm:impossibility} extends to the case where the shared randomness is completely known to the server; this can be shown by a standard averaging argument.

A natural way to bypass this impossibility is to provide a key $K_i$ to client $i$ such that $(K_1,\ldots,K_n)$ is distributed according to some publicly known joint distribution $P_{K_1\cdots K_n}$, and the server is allowed to collude with up to $t$ clients. This reveals the key and sample of each colluding client $i$ to the server.
Since the server gets $t$ independent samples from the test distribution from such collusions, the privacy error is lower bounded by the statistical distance between $t$ samples from any pair of distributions in $\cH_0$ (or $\cH_1$).
Nevertheless, when $t$ is small, it may be possible to design SDHT schemes with desirable correctness and privacy guarantees.
A good candiate for constructing such tests is the notion of \emph{non-interactive secure multiparty computation} (NIMPC) introduced by Beimel et al.~\cite{BGIKMP14}.

\IEEEtriggeratref{16}

\bibliographystyle{IEEEtran}
\bibliography{bibliofile}

\clearpage
\if \extended 1%
\appendix
\section*{Proof of \Cref{prop:lp}}
    Let $\det:\prod_{i=1}^n\cY_i \to \{0,1\}$ be a $(1-\epsilon)$-correct hypothesis test for $P^{(0)}_Y$ vs. $P^{(1)}_Y$ using $n$ samples for appropriate $\epsilon>0$. Then, it is easy to see that $\left(\{W_{Y_i|X_i}\}_{i=1}^n,\det\right)$ where $W_{X_i|Y_i}=W_{X|Y}$ for each $i$ is a $(\epsilon,0)$ private distributed hypothesis test for $\cH_0$ vs. $\cH_1$ by the property of $W_{Y|X}$. An optimal hypothesis test $\det$ will achieve correctness error $\epsilon = e^{-\Omega(n)}$.    

\section*{Missing Details in Proof of Theorem~\ref{thm:impossibility}}
We first prove Proposition~\ref{prop:propFiniteRatio} and formalize the arguments outlined in items \textit{(i)-(iii)} that appear after it.
\subsection*{Proof of Proposition~\ref{prop:propFiniteRatio}}
For a channel $W_{Y|X}$ with $W_{Y|X}(1|0) = q_0$ and $W_{Y|X}(1|1) = q_1,~q_0 > q_1$, WLOG, we can write
    \begin{align*}
        P_{Y}^{(2)} & \equiv \Ber(a),\,P_{Y}^{(1)} \equiv \Ber(a + c'),\,
        P_{Y}^{(0)}  \equiv \Ber(a+c),
    \end{align*}
    where $a = q_1,\, c = q_0-q_1,$ and $c' = (1-\theta)c$. Define
    \begin{align*}
        f(a,c) & := \frac{\hd^2\left(P^{(1)}_{Y},P^{(2)}_{Y}\right)}{\hd^2\left(P^{(0)}_{Y},P^{(1)}_{Y}\right)}                                        \\
               = &\frac{1-\sqrt{a(a+c')}-\sqrt{(1-a)(1-a-c')}}{1-\sqrt{(a+c')(a+c)}-\sqrt{(1-a-c')(1-a-c)}}\\
               :=&\frac{N}{D}.
    \end{align*}
    We seek to show that for all choices of (non-degenerate) $W_{Y|X}$, the ratio defined by $f(a,c)$ is bounded. We handle a boundary case first. If $a+c=1$,  then $q_0 = 1$ and $c = 1-q_1$. For this configuration, we have
\begin{align*}
        &f(a,c) \\
        ={}&\frac{1-\sqrt{a(a+c')}-\sqrt{(1-a)(1-a-c')}}{1-\sqrt{a+c'}}      \\
        ={}&\frac{1-\sqrt{(1-c)(1-c+c')}-\sqrt{c(c-c')}}{1-\sqrt{1-c+c'}}    \\
        ={}&\frac{1-\sqrt{(1-c)(1-c+(1-\theta)c)}-\sqrt{c(c-(1-\theta)c)}}{1-\sqrt{1-c+(1-\theta)c}} \\
        ={}&\frac{1-\sqrt{(1-c)(1-\theta c)}-\sqrt{c^2\theta}}{1-\sqrt{1-\theta c}}  \\
        ={}& \sqrt{1-c} + \frac{1-\sqrt{1-c}-c\sqrt{\theta}}{1-\sqrt{1-\theta c}}\\
        \le{}& \sqrt{1-c} + \frac{(1-\sqrt{1-c}-c\sqrt{\theta})(1+\sqrt{1-\theta c})}{\theta c}\\
        \le{}&\sqrt{1-c} + \frac{(c-c\sqrt{\theta})\cdot2}{\theta c}\\
        \le{}& 1 + \frac{2(1-\sqrt{\theta})}{\theta}
    \end{align*}
    which is bounded above for a fixed $\theta \in (0,1)$.

    Next, we show that $f(a,c)$ is non-increasing in $a \in (0,1-c)$ for a fixed $c \in (0,1)$, so that we can bound $f(a,c)$ by $f(0,c)$.  The partial derivatives w.r.t. $a$ are given by
    \begin{align*}
        \frac{\partial f}{\partial a} & = \frac{\partial N}{\partial a} \cdot D - \frac{\partial D}{\partial a} \cdot N.
    \end{align*}
    To show that $f(a,c)$ is non-increasing, we show that
    \[
     \frac{\partial f}{\partial a} \le 0 \Longleftrightarrow \frac{\partial N}{\partial a} \cdot \frac{1}{N} \le \frac{\partial D}{\partial a} \cdot \frac{1}{D}.
    \]
   Define, for $b \neq c'$,
\begin{align}
    &g_{a,c'}(b) \nonumber \\
    &= \frac{\frac{\partial }{\partial a}\left(1-\sqrt{(a+b)(a+c')}-\sqrt{(1-a-b)(1-a-c')}\right)}{1-\sqrt{(a+b)(a+c')}-\sqrt{(1-a-b)(1-a-c')}} \nonumber \\
    &= \left( 1-\sqrt{(a+b)(a+c')} - \sqrt{(1-a-b)(1-a-c')} \right)^{-1} \nonumber \\
    &\quad \times \left( -\frac{2a+b+c'}{2\sqrt{(a+b)(a+c')}} + \frac{2-2a-b-c'}{2\sqrt{(1-a-b)(1-a-c')}} \right). \nonumber
\end{align}
Substituting $p = a+b$ and $q = a+c'$, the expression simplifies to
\begin{align}
    &g_{a,c'}(b) \nonumber\\
    &= \frac{-\frac{p+q}{2\sqrt{pq}}+\frac{2-p-q}{2\sqrt{(1-p)(1-q)}}}{1-\sqrt{pq}-\sqrt{(1-p)(1-q)}} \nonumber \\
    &= \frac{(\sqrt{(1-p)(1-q)}-\sqrt{pq})(1+\sqrt{(1-p)(1-q)}+\sqrt{pq})}{-2\sqrt{pq(1-p)(1-q)}} \label{eqn:simplificationprop} \\
    &:= h_q(p), \nonumber
\end{align}
where the equality in \eqref{eqn:simplificationprop} is proved in Lemma~\ref{lemma1} towards the end of this section on page~\pageref{lemma1}.

Note that $h_q(p)$ is well-defined for the case $p=q$ (that is, $b = c'$) as well. To show $g_{a,c'}(0)\leq g_{a,c'}(c)$, it suffices to prove that $h_{q}(p)$ is monotonically increasing in $p$ for a fixed $q$; that is $\frac{\partial h_q}{\partial p} \ge 0$.

Let $p=\sin^2\alpha$, $q=\sin^2\beta$ where $\alpha \in [0,\frac{\pi}{2}]$. The case where $\alpha = 0$ corresponds to $p = a + b = 0$ for  $b \in [0,c]$, in particular, $a = 0$, and this case is handled separately, later. The case where $\alpha = \pi/2$ corresponds to $p = a + b = 1$ where $b \in [0,c]$. This implies that $a + c = 1$, a case that has been handled separately earlier. So, we restrict to $\alpha \in \left(0,\frac{\pi}{2}\right)$.

To show that $h_{q}(p)$ is monotonically increasing in $p$ for a fixed $q$, it suffices to show that for the function $$f_\beta(\alpha) := h_{\sin^2\beta}(\sin^2\alpha),$$ the partial derivative $\frac{\partial f_\beta}{\partial \alpha}\geq 0$, since
\[
\frac{\partial h_q}{\partial p}=\frac{\partial f_\beta}{\partial \alpha} \frac{\partial \alpha}{\partial p}=\frac{\partial f_\beta}{\partial \alpha}\frac{1}{\sin{2\alpha}}
\]
and $\sin{2\alpha}\geq 0$ for $\alpha\in\left(0,\frac{\pi}{2}\right)$. We have
\begin{align*}
        f_\beta(\alpha) 
        ={}& \frac{-1}{2}\frac{(\sqrt{(1-\sin^2\alpha)(1-\sin^2\beta)}-\sqrt{\sin^2\alpha \sin^2\beta})}{\sqrt{\sin^2\alpha\sin^2\beta(1-\sin^2\alpha)(1-\sin^2\beta)}} \\
        &{}\cdot (1+\sqrt{(1-\sin^2\alpha)(1-\sin^2\beta)}+\sqrt{\sin^2\alpha \sin^2\beta})\\
        ={}&\frac{-2\cos{(\alpha+\beta)}(1+\cos{(\alpha-\beta)})}{\sin{2\alpha}\sin{2\beta}}.
    \end{align*}
The partial derivative of $f_\beta$ with respect to $\alpha$ satisfies
\begin{align*}
        &{}\frac{1}{2}\sin^2{2\alpha}\sin{2\beta}\frac{\partial f_\beta}{\partial \alpha}\\
        ={}&-\sin{2\alpha}\Big(-\sin{(\alpha+\beta)}-\sin{2\alpha})+2\cos{2\alpha}(\cos{(\alpha+\beta)}\\
        &{}+\frac{1}{2}(\cos{2\alpha}+\cos{2\beta})\Big)                    \\
        ={}&\sin{2\alpha}\sin{(\alpha+\beta)}+\sin^2{2\alpha}+2\cos{2\alpha}\cos{(\alpha+\beta)}\\
        &{}+\cos{2\alpha}(\cos{2\alpha}+\cos{2\beta})                      \\
         ={}&\cos{(\alpha-\beta)}-\cos{2\alpha}\cos{(\alpha+\beta)}+\sin^2{2\alpha}\\
         &{}+2\cos{2\alpha}\cos{(\alpha+\beta)}+\cos^2{2\alpha}+\cos{2\alpha}\cos{2\beta} \\
         ={}&\cos{(\alpha-\beta)}+\cos{2\alpha}\cos{(\alpha+\beta)}+\sin^2{2\alpha}\\
         &{}+\cos^2{2\alpha}+\cos{2\alpha}\cos{2\beta}                                    \\
         ={}&1+\cos{(\alpha-\beta)}+\cos{2\alpha}\cos{2\beta}+\cos{2\alpha}\cos{(\alpha+\beta)}                                                                  \\
         ={}&(1+\cos{2\alpha}\cos{2\beta})+(\cos{(\alpha-\beta)}+\cos{2\alpha}\cos{(\alpha+\beta)}).
    \end{align*}
    The first term $1+\cos{2\alpha}\cos{2\beta}\geq 0$, and we show that the second term is also non-negative:
    \begin{align*}
        &{}\cos{(\alpha-\beta)}+\cos{2\alpha}\cos{(\alpha+\beta)}\\
         ={}&(\cos{\alpha}\cos{\beta}+\sin{\alpha}\sin{\beta})\\
         &{}+(2\cos^2{\alpha}-1)(\cos{\alpha}\cos{\beta}-\sin{\alpha}\sin{\beta})                    \\                                          ={}&\cancel{\cos{\alpha}\cos{\beta}}+\sin{\alpha}\sin{\beta}+2\cos^3{\alpha}\cos{\beta}\\
         &{}-2\cos^2{\alpha}\sin{\alpha}\sin{\beta}\cancel{-\cos{\alpha}\cos{\beta}}+\sin{\alpha}\sin{\beta} \\
         ={}&2\sin{\alpha}\sin{\beta}+2\cos^3{\alpha}\cos{\beta}-2\cos^2{\alpha}\sin{\alpha}\sin{\beta}                                                                                          \\
         ={}&2\cos^3{\alpha}\cos{\beta}+2\sin{\alpha}\sin{\beta}(1-\cos^2{\alpha})                                                                                                               \\
         ={}&2\cos^3{\alpha}\cos{\beta}+2\sin^3{\alpha}\sin{\beta}                                                                                                                               \\
         \geq{}& 0,
    \end{align*}
    where the last inequality holds because $\sin{\alpha},\cos{\alpha}\geq 0$ for $\alpha\in(0,\frac{\pi}{2})$. Thus, we have established that $f(a,c)$ is non-increasing in $a \in (0,1-c)$ for a fixed $c \in (0,1)$ and hence
    \begin{align}
        f(a,c) &\le f(0,c)\nonumber\\
        &=\frac{1-\sqrt{c'}-\sqrt{1-c'}}{1-\sqrt{c'c}-\sqrt{(1-c')(1-c)}}\nonumber\\
        &=\frac{1-\sqrt{(1-\theta)c}-\sqrt{1-(1-\theta)c}}{1-\sqrt{(1-\theta)c^2}-\sqrt{(1-(1-\theta)c)(1-c)}}\nonumber\\
        & =\frac{\sqrt{k}-c+\sqrt{(1-c)(k-c)}}{(\sqrt{k}+\sqrt{k-c})(\sqrt{k}-1)^2}\label{eqn:f(0,c)}
    \end{align}
    where $k = \frac{1}{1-\theta} > 1$. In Lemma~\ref{lemma2} towards the end of this section on page~\pageref{lemma2}, we show that $f(0,c)$ is non-increasing in $c \in (0,1)$. 
    Therefore,
    \begin{align*}
        f(a,c) \le f(0,c) \le \lim_{c \downarrow0}f(0,c) =\frac{1}{\left(\sqrt{k}-1\right)^2} < \infty.
    \end{align*}
\hfill $\square$

\subsection*{Formalizing the Arguments in \textit{(i)-(iii)}}
{\em{(i)}} Let $W_{Y|X}$ be a channel with input alphabet $\mathcal{X}=\{0,1\}$ and output alphabet $\mathcal{Y}=\{1,2,\ldots,k\}$ such that, for $i\in[1:k]$,
\begin{align*}
    W_{Y|X}(i|0) = \alpha^{(i)}_0 \text{ and} \
    W_{Y|X}(\cdot|1) = \alpha^{(i)}_1, 
\end{align*}
where 
\[
    \frac{\alpha^{(i)}_0}{\alpha^{(i)}_1} \le \frac{\alpha^{(j)}_0}{\alpha^{(j)}_1} \quad \forall i\le j. 
\]

We consider the transformation $W_{Y|X}\rightarrow W^{\gamma}_{Y|X}$ for $\gamma\in(0,1)$ such that, for $b\in\{0,1\}$,
\begin{align*}
    &W^{\gamma}_{Y|X}(i|b) = \\
    &\begin{cases}
        (1-\gamma)\cdot \alpha^{(1)}_b + \gamma & \text{for } i=1         \\
        (1-\gamma)\cdot \alpha^{(i)}_b          & \text{for } i=2,\ldots,n
    \end{cases} 
\end{align*}
We have the following lemma proved towards the end of this section on page~\pageref{proofoflemma1}.

\begin{lemma}\label{claim}
   Suppose $\frac{\alpha^{(1)}_0}{\alpha^{(1)}_1} < \frac{\alpha^{(2)}_0}{\alpha^{(2)}_1}\leq 1$. There exists $\gamma^*\in(0,1)$ such that 
    \begin{align*}
        \text{i.~}& \frac{(1-\gamma^*)\cdot \alpha_0^{(1)}+\gamma^*}{(1-\gamma^*)\cdot \alpha_1^{(1)}+\gamma^*} = \frac{\alpha_0^{(2)}}{\alpha_1^{(2)}}, \text{ and}\\
        \text{ii.~}& \frac{(1-\gamma)\cdot \alpha_0^{(1)}+\gamma}{(1-\gamma)\cdot \alpha_1^{(1)}+\gamma}\le \frac{\alpha_0^{(2)}}{\alpha_1^{(2)}}
    \end{align*}
    for $\gamma\in[0,\gamma^*]$.
\end{lemma}
{\em{(ii)}} We consider input distributions $\mu^{0}=\mathrm{Ber}(0), \mu^{1}=\mathrm{Ber}(\lambda), \mu^{2}=\mathrm{Ber}(1)$.
Let $(a^{(j)}_i)_{j=1}^k$ be the output distribution of $W$ when the input distribution is $\mu^{i}$, for $i=0,1,2$. For channel $W$, we are interested in the ratio
\[ H = \frac{\sum_{j=1}^k\left(\sqrt{a^{(j)}_2} - \sqrt{a^{(j)}_1}\right)^2 }{\sum_{j=1}^k\left(\sqrt{a^{(j)}_0} - \sqrt{a^{(j)}_1}\right)^2}.\]
For the channel $W^{\gamma}$, the related ratio $H(\gamma)$ is given in Equation \eqref{eq:H_gamma}.

\begin{figure*}[t]
\normalsize
\begin{align}
    H(\gamma) 
    &= \frac{\left(\sqrt{(1-\gamma) a^{(1)}_2 + \gamma} - \sqrt{(1-\gamma) a^{(1)}_1 + \gamma}\right)^2 + \sum_{j=2}^n \left(\sqrt{(1-\gamma) a^{(j)}_2} - \sqrt{(1-\gamma) a^{(j)}_1}\right)^2}
    {\left(\sqrt{(1-\gamma) a^{(1)}_0 + \gamma} - \sqrt{(1-\gamma) a^{(1)}_1 + \gamma}\right)^2 + \sum_{j=2}^n \left(\sqrt{(1-\gamma) a^{(j)}_0} - \sqrt{(1-\gamma) a^{(j)}_1}\right)^2} \nonumber \\[3ex]
    &= \frac{\left(\sqrt{a^{(1)}_2 + \frac{\gamma}{1-\gamma}} - \sqrt{a^{(1)}_1 + \frac{\gamma}{1-\gamma}}\right)^2 + \sum_{j=2}^n \left(\sqrt{a^{(j)}_2} - \sqrt{a^{(j)}_1}\right)^2}
    {\left(\sqrt{a^{(1)}_0 + \frac{\gamma}{1-\gamma}} - \sqrt{a^{(1)}_1 + \frac{\gamma}{1-\gamma}}\right)^2 + \sum_{j=2}^n \left(\sqrt{a^{(j)}_0} - \sqrt{a^{(j)}_1}\right)^2}. \label{eq:H_gamma}
\end{align}
\vspace*{2pt}
\hrulefill
\vspace*{4pt}
\end{figure*}

We show in Lemma~\ref{lemma:Hinc} towards the end of this section on page~\pageref{lemma:Hinc} that $H(\gamma)$ is strictly decreasing for $\gamma\in[0,\gamma^*]$.

{\em{(iii)}} Finally, to show that the squared Hellinger ratio does not change if two output symbols $i,j$ such that $\frac{\alpha_0^{(i)}}{\alpha_1^{(i)}}=\frac{\alpha_0^{(j)}}{\alpha_1^{(j)}}$, it suffices to show the following lemma, which is proved towards the end of this section on page~\pageref{proofoflemma2}.

\begin{lemma}\label{lemma:merge}
    Let $X_1 \sim \mathrm{Ber}(p)$ and $X_2 \sim \mathrm{Ber}(q)$ be passed through a channel $W_{Y|X}$ with output alphabet $\mathcal{Y}$. Let $i, j \in \mathcal{Y}$ be two distinct output symbols with
    \begin{align*}
        W_{Y|X}(i|0) &= \alpha_0^{(i)}, & W_{Y|X}(i|1) &= \alpha_1^{(i)}, \\
        W_{Y|X}(j|0) &= \alpha_0^{(j)}, & W_{Y|X}(j|1) &= \alpha_1^{(j)}.
    \end{align*}
    Let $P_Y^{(1)}$ and $P_Y^{(2)}$ denote the respective output distributions. Suppose a second channel merges symbols $i$ and $j$ into a supersymbol $ij$, such that the new output $Y^\prime$ is defined as
    \[
        Y^\prime = \begin{cases}
            ij & \text{if } Y \in \{i, j\} \\
            Y & \text{otherwise}
        \end{cases}.
    \]
    Let $P_{Y^\prime}^{(1)}$ and $P_{Y^\prime}^{(2)}$ denote the output distributions after merging. If $\frac{\alpha_0^{(i)}}{\alpha_1^{(i)}} = \frac{\alpha_0^{(j)}}{\alpha_1^{(j)}}$, then for any $p, q \in [0, 1]$,
    \[
        H^2\left(P_Y^{(1)}, P_Y^{(2)}\right) =  H^2\left(P_{Y^\prime}^{(1)}, P_{Y^\prime}^{(2)}\right).
    \]
\end{lemma}
We can repeat the process in {\em{(i)-(iii)}} to merge the symbols in $\mathcal{Y}$ until we obtain a channel $W^\prime_{Y|X}$ such that there is a single output symbol $a$ such that $W^{\prime}_{Y|X}(a|0)\leq W^{\prime}_{Y|X}(a|1)$. Relabeling $0$ and $1$, we can continue this process to obtain a binary channel $W^{\prime\prime}_{Y|X}$.

\subsection*{Proofs of the lemmas used in Proof of Theorem~\ref{thm:impossibility}}
\begin{lemma}\label{lemma1}
    Let $0 < p < 1$, $0 < q < 1$, and $p \neq q$. Define the expression $h_q(p)$ as
    \begin{equation}
        h_q(p) = \frac{-\frac{p+q}{2\sqrt{pq}}+\frac{2-p-q}{2\sqrt{(1-p)(1-q)}}}{1-\sqrt{pq}-\sqrt{(1-p)(1-q)}}
    \end{equation}
    The function $h_q(p)$ can be simplified to the following form:
    \begin{align}
        h_q(p) &= -\frac{1}{2\sqrt{pq(1-p)(1-q)}} \nonumber \\
        &\quad \times \left(\sqrt{(1-p)(1-q)}-\sqrt{pq}\right) \nonumber \\
        &\quad \times \left(1+\sqrt{(1-p)(1-q)}+\sqrt{pq}\right)
    \end{align}
\end{lemma}

\begin{proof}
    Define $A=\sqrt{pq}$ and $B=\sqrt{(1-p)(1-q)}$. 
    First, we combine the terms in the numerator of $h_q(p)$.
    \begin{equation*}
        -\frac{p+q}{2A}+\frac{2-p-q}{2B}=\frac{A(2-p-q)-B(p+q)}{2AB}.
    \end{equation*}
    We have
    \begin{align*}
        A^2 &= pq, \\
        B^2 &= (1-p)(1-q) = 1-p-q+pq,
    \end{align*}
    which gives
    \begin{align*}
        1-p-q &= B^2-A^2, \\
        2-p-q &= 1+B^2-A^2, \\
        p+q &= 1-B^2+A^2.
    \end{align*}
    Substituting these into the numerator yields
    \begin{align*}
        A(2-p-q) &- B(p+q) \\
        &= A(1+B^2-A^2)-B(1-B^2+A^2) \\
        &= A+AB^2-A^3-B+B^3-A^2B \\
        &= (A-B)+(AB^2-A^2B)+(B^3-A^3).
    \end{align*}
    We factor the terms individually.
    \begin{align*}
        AB^2-A^2B &= AB(B-A) = -(A-B)AB, \\
        B^3-A^3 &= (B-A)(B^2+AB+A^2) \\
        &= -(A-B)(A^2+AB+B^2).
    \end{align*}
    Substituting these back into the numerator, we get
    \begin{align*}
         &A(2-p-q) - B(p+q)\\
         &= (A-B)\left[1-AB-(A^2+AB+B^2)\right] \\
        &= (A-B)\left(1-A^2-B^2-2AB\right) \\
        &= (A-B)\left(1-(A+B)^2\right) \\
        &= (A-B)(1-A-B)(1+A+B).
    \end{align*}
    Substituting this back into the expression for $h_q(p)$, we have
   \begin{align}
        h_q(p) &= \frac{ \frac{(A-B)(1-A-B)(1+A+B)}{2AB} }{1-A-B}\nonumber\\
         &= \frac{(A-B)(1+A+B)}{2AB}\label{eqn:pnotq} \\
        &= -\frac{1}{2}\frac{(B-A)(1+A+B)}{AB}\nonumber,
    \end{align}
    where \eqref{eqn:pnotq} follows because $p\neq q$ implies that $A+B\neq 1$. 
    Finally, substituting the expressions for $A$ and $B$ back into this yields
    \begin{align*}
        h_q(p) &= -\frac{1}{2\sqrt{pq(1-p)(1-q)}} \\
        &\quad \times \left(\sqrt{(1-p)(1-q)}-\sqrt{pq}\right) \\
        &\quad \times \left(1+\sqrt{(1-p)(1-q)}+\sqrt{pq}\right).
    \end{align*}
\end{proof}

\begin{lemma}\label{lemma2}
The function $f(0,c)$ in \eqref{eqn:f(0,c)} is non-increasing on the interval $(0, 1)$.
\end{lemma}

\begin{proof}

     The derivative of $f(0,c)$ w.r.t. $c$ evaluates to
    \[
    f'(0,c)=\frac{-r(q+r)(r^2+2s-c+1)+s(q+s-c)}{2rs(q-1)^2(q+r)^2},
    \]
    where $$q=\sqrt{k},~r=\sqrt{k-c},~ s=\sqrt{(1-c)(k-c)}.$$
    The denominator is positive, and the sign of $f'(0,c)$ is determined entirely by its numerator,
\[
\mathrm{Num} = -r(q+r)(r^2+2s-c+1) + s(q+s-c).
\]
Define $t = \sqrt{1-c}$. We have
\begin{enumerate}
    \item $c = 1 - t^2$,
    \item $s = \sqrt{(1-c)(k-c)} = tr$,
    \item $r^2 - t^2 = (k-c) - (1-c) = k-1 = q^2 - 1$.
\end{enumerate}
Substituting $c = 1-t^2$ and $s = tr$ into the components of the terms in the numerator yields
\begin{align*} 
r^2+2s-c+1 &= r^2+2tr-(1-t^2)+1 \\
&= r^2+2tr+t^2 = (r+t)^2, 
\end{align*}
\begin{align*} 
q+s-c &= q+tr-(1-t^2) \\
&= q-1+t^2+tr = (q-1)+t(r+t). 
\end{align*}
Applying these substitutions back into $\mathrm{Num}$ gives
\[
\mathrm{Num} = -r(q+r)(r+t)^2 + tr \big[ (q-1) + t(r+t) \big],
\]
which implies
\begin{align*}
    &\frac{\mathrm{Num}}{r} \\
    &= -(q+r)(r+t)^2 + t(q-1) + t^2(r+t) \\
&= t(q-1) - (r+t) \big[ (q+r)(r+t) - t^2 \big] \\
&= t(q-1) - (r+t) \big[ q(r+t) + rt + (r^2 - t^2) \big] \\
&= t(q-1) - (r+t) \big[ q(r+t) + rt + (q^2 - 1) \big] \\
&= t(q-1) - (r+t)(q^2-1) - q(r+t)^2 - rt(r+t) \\
&= (q-1) \big[ t - (r+t)(q+1) \big] - q(r+t)^2 - rt(r+t) \\
&= (q-1) [t - rq - r - tq - t] - q(r+t)^2 - rt(r+t) \\
&= -(q-1)(rq + r + tq) - q(r+t)^2 - rt(r+t).
\end{align*}

Noting that $q=\sqrt{k}>1$, we have that all the three terms above are negative, and hence $f'(0,c)\leq 0$. Thus, $f(0,c)$ is non-decreasing in $(0,1)$.
\end{proof}

\begin{proof}[Proof of Lemma~\ref{claim}]\label{proofoflemma1}
    For notational simplicity, we denote $\alpha^{(1)}_i$ by $a_i$ and $\alpha^{(2)}_i$ by $b_i$, for $i\in\{0,1\}$.
    First, we show that $f'(\gamma)\ge 0$ where
    \[ 
        f(\gamma) = \frac{(1-\gamma)a_0 + \gamma}{(1-\gamma)a_1 + \gamma} = \frac{a_0 + \gamma(1-a_0)}{a_1 + \gamma(1-a_1)}.
    \]

    To determine its monotonicity, we compute  $f'(\gamma)$.
    \begin{align*}
        f'(\gamma) &= \frac{1}{[a_1 + \gamma(1-a_1)]^2} \Big( (1-a_0)[a_1 + \gamma(1-a_1)] \\
        &\quad - (1-a_1)[a_0 + \gamma(1-a_0)] \Big).
    \end{align*}

    Expanding the terms inside the numerator yields
    \begin{align*}
        &(a_1 + \gamma - a_1 \gamma - a_0 a_1 - a_0 \gamma + a_0 a_1 \gamma) \\
        &- (a_0 + \gamma - a_0 \gamma - a_1 a_0 - a_1 \gamma + a_1 a_0 \gamma).
    \end{align*}

    Notice that all terms involving $\gamma$ cancel out, simplifying the numerator to $a_1 - a_0$. Thus, the derivative is
    \[ 
        f'(\gamma) = \frac{a_1 - a_0}{[a_1 + \gamma(1-a_1)]^2}.
    \]

    From the premise $\frac{a_0}{a_1} < \frac{b_0}{b_1} \le 1$, we have $\frac{a_0}{a_1} < 1$, which implies $a_0 < a_1$ (since the variables are non-negative). Therefore, the numerator $a_1 - a_0 > 0$. Since the denominator is strictly positive, $f'(\gamma) > 0$. This confirms that the ratio strictly increases as $\gamma$ increases.

    Next, we show that the $\gamma$ satisfying $f(\gamma) = \frac{b_0}{b_1}$ falls in the interval $[0, 1]$. We set up the equation
    \[ 
        \frac{a_0 + \gamma(1-a_0)}{a_1 + \gamma(1-a_1)} = \frac{b_0}{b_1}.
    \]

    This gives
     \[ 
        \gamma = \frac{N}{M + N},
    \]
    where $N = a_1 b_0 - a_0 b_1$ and $M = b_1 - b_0$.
    We analyze $N$ and $M$ based on the given conditions.
    \begin{enumerate}
        \item Since $\frac{a_0}{a_1} < \frac{b_0}{b_1}$, cross-multiplying yields $a_0 b_1 < a_1 b_0$, which means $N = a_1 b_0 - a_0 b_1 > 0$.
        \item Since $\frac{b_0}{b_1} \le 1$, we have $b_0 \le b_1$, which means $M = b_1 - b_0 \ge 0$.
    \end{enumerate}

    Because $N > 0$ and $M \ge 0$, the denominator $M + N$ is strictly positive and $M + N \ge N$. Consequently, the fraction $\frac{N}{M+N}$ must lie in the interval $(0, 1]$. This satisfies the required condition $0 \le \gamma \le 1$.
\end{proof}

\begin{lemma}\label{lemma:Hinc}
    $H(\gamma)$ in \eqref{eq:H_gamma} is strictly increasing for $\gamma\in[0,\gamma^*]$. 
\end{lemma}
\begin{proof}

    \emph{\underline{Step 1:}}
Let $x = \frac{\gamma}{1-\gamma}$. Because $x$ is a strictly increasing function for $\gamma \in [0, 1)$, $H(\gamma)$ is strictly increasing if and only if $H(x)$ is strictly increasing. From Lemma~\ref{claim},
\[
 \frac{a_0^{(1)} + x}{a_2^{(1)} + x} \le \frac{a_0^{(2)}}{a_2^{(2)}}.
\]

Because $a_0^{(1)} < a_2^{(1)}$, the function $x \mapsto \frac{a_0^{(1)}+x}{a_2^{(1)}+x}$ is strictly increasing in $x$. If $\frac{a_0^{(1)}}{a_2^{(1)}} \ge \frac{a_0^{(2)}}{a_2^{(2)}}$, the constraint is already tight or violated at $x=0$. In this case, the admissible domain is only the singleton $\{0\}$, where the theorem holds vacuously.

We therefore proceed assuming $x^* > 0$, which strictly requires $\frac{a_0^{(1)}}{a_2^{(1)}} < \frac{a_0^{(2)}}{a_2^{(2)}}$. Since probabilities are non-negative, this strictly implies $a_0^{(2)} > 0$.

We rewrite $H(x) = \frac{G(x) + C_2}{F(x) + C_0}$, where
\begin{itemize}
    \item $F(x) = \left(\sqrt{a_0^{(1)}+x} - \sqrt{a_1^{(1)}+x}\right)^2$,
    \item $G(x) = \left(\sqrt{a_2^{(1)}+x} - \sqrt{a_1^{(1)}+x}\right)^2$,
    \item $C_0 = \sum_{j=2}^n \left(\sqrt{a_0^{(j)}} - \sqrt{a_1^{(j)}}\right)^2 = \sum_{j=2}^n c_{0,j}$,
    \item $C_2 = \sum_{j=2}^n \left(\sqrt{a_2^{(j)}} - \sqrt{a_1^{(j)}}\right)^2 = \sum_{j=2}^n c_{2,j}$.
\end{itemize}

\emph{\underline{Step 2:}}
We differentiate $H(x)$ to get
\[
H'(x) = \frac{G'(x)(F(x)+C_0) - F'(x)(G(x)+C_2)}{(F(x)+C_0)^2}.
\]
Since $F(x)$ and $G(x)$ strictly decrease as $x$ increases, $F'(x) = -|F'(x)|$ and $G'(x) = -|G'(x)|$. For $H'(x) \ge 0$ to hold, we require
\[
-|G'(x)|(F(x)+C_0) + |F'(x)|(G(x)+C_2) \ge 0,
\]
which directly simplifies to the condition
\[
|F'(x)|(G(x)+C_2) \ge |G'(x)|(F(x)+C_0).
\]

Also,
\begin{align*}
    |F'(x)| &= F(x) \cdot \frac{1}{\sqrt{(a_0^{(1)}+x)(a_1^{(1)}+x)}}, \\
    |G'(x)| &= G(x) \cdot \frac{1}{\sqrt{(a_2^{(1)}+x)(a_1^{(1)}+x)}}.
\end{align*}

Because $a_0^{(1)} < a_2^{(1)}$, it strictly follows that $\frac{1}{\sqrt{a_0^{(1)}+x}} > \frac{1}{\sqrt{a_2^{(1)}+x}}$. So, the following inequality is sufficient to imply $H'(x) > 0$.
\[
\frac{F(x)}{G(x)} \ge \frac{C_0}{C_2}.
\]

\emph{\underline{Step 3:}}
For $t > 0$, define the function
\[
K(t) = \frac{\left(1 - \sqrt{\theta + (1-\theta)t}\right)^2}{\left(\sqrt{t} - \sqrt{\theta + (1-\theta)t}\right)^2}.
\]

Let $y = \sqrt{t}$, $q = 1-\theta$, and $w = \sqrt{\theta + qy^2}$. For $y \neq 1$, define
\[
M(y) = \left|\frac{w-1}{y-w}\right|.
\]
Since $K(t) = M(y)^2$, it is enough to prove that $M(y)$ is increasing in $y$. For $y > 1$, we have $M(y) = \frac{w-1}{y-w}$. Differentiation gives
\[
\frac{M'(y)}{M(y)} = \frac{w'}{w-1} - \frac{1-w'}{y-w}.
\]
Thus $M'(y) > 0$ is equivalent to $w'(y-w) > (1-w')(w-1)$. Using $w' = \frac{qy}{w}$, this reduces to $w > \lambda + qy$. But
\[
w^2 - (\theta+qy)^2 = \lambda q(y-1)^2 > 0.
\]
Hence $M'(y) > 0$ for $y > 1$. Differentiating $M(y) = \frac{1-w}{w-y}$ for $0 < y < 1$ similarly gives the exact same equivalent condition $w > \theta + qy$, proving $K(t)$ is strictly increasing for all $t > 0, t \neq 1$.

\emph{\underline{Step 4:}}

 Since we are in the case $a_0^{(2)}>0$, the monotone likelihood ratio condition implies that, for every $j\ge 2$,
\[
a_0^{(2)}a_2^{(j)} \le a_0^{(j)}a_2^{(2)}.
\]
Since $a_2^{(2)}>0$, if $a_0^{(j)}=0$, then necessarily $a_2^{(j)}=0$. Such an index contributes zero to both $C_0$ and $C_2$, and may be omitted.

Thus, for every active index $j\ge 2$, we have $a_0^{(j)}>0$ and may define $t_j=\frac{a_2^{(j)}}{a_0^{(j)}}$.

If $t_j = 1$, then $a_0^{(j)} = a_2^{(j)}$, meaning both $c_{0,j}$ and $c_{2,j}$ vanish. We resolve this by defining $K(1)$ through its continuous extension,
\[
K(1) := \lim_{t\to 1}K(t) = \left(\frac{1-\lambda}{\lambda}\right)^2.
\]
With this convention, we have
\[
c_{0,j} = c_{2,j}K(t_j),
\]
which remains valid for all active $j \ge 2$, including when $t_j = 1$.

\emph{\underline{Step 5:}}
For $x \in (0, x^*]$, the term $a_0^{(1)}+x$ is strictly positive. Now we identify
\[
\frac{F(x)}{G(x)} = K(t_1(x)), \qquad t_1(x) = \frac{a_2^{(1)}+x}{a_0^{(1)}+x}.
\]
By the monotone likelihood ratio condition, the ratios $t_j=\frac{a_2^{(j)}}{a_0^{(j)}}$ are non-increasing over the active indices $j\ge 2$. In particular, $t_2\ge t_j$ for every active $j\ge 2$.

Recall that, for $x \in (0, x^*]$,
\[
\frac{a_0^{(1)}+x}{a_2^{(1)}+x} \le \frac{a_0^{(2)}}{a_2^{(2)}}.
\]
Since all quantities in this inequality are strictly positive for $x > 0$ in the nondegenerate case, this is equivalent to
\[
t_1(x) = \frac{a_2^{(1)}+x}{a_0^{(1)}+x} \ge \frac{a_2^{(2)}}{a_0^{(2)}} = t_2.
\]
Therefore, $t_1(x)\ge t_j$ for every active $j\ge 2$. Since $K$ is strictly increasing,
\[
K(t_1(x))\ge K(t_j).
\]
Hence
\begin{align*}
    C_0 &= \sum_{j=2}^n c_{0,j} = \sum_{j=2}^n c_{2,j}K(t_j) \\
    &\le K(t_1(x))\sum_{j=2}^n c_{2,j} = K(t_1(x))C_2.
\end{align*}
Thus,
\[
\frac{C_0}{C_2} \le K(t_1(x)) = \frac{F(x)}{G(x)}.
\]

From Step 2, this implies that $H'(x) > 0$ for all $x \in (0, x^*]$.

Because $a_1$ is a strict convex combination of $a_0$ and $a_2$, the inequality $a_0^{(1)} < a_2^{(1)}$ guarantees that $a_0^{(1)} < a_1^{(1)}$. Consequently, $F(0) = \left(\sqrt{a_0^{(1)}} - \sqrt{a_1^{(1)}}\right)^2 > 0$. The denominator of $H(x)$, which is $F(x) + C_0$, is therefore strictly positive and non-zero. This makes $H(x)$ continuous on the closed interval $[0, x^*]$. Since $H(x)$ is continuous on $[0, x^*]$ and strictly increasing on $(0, x^*]$, it is strictly increasing on $[0, x^*]$. \hfill 
\end{proof}

    \begin{proof}[Proof of Lemma~\ref{lemma:merge}]\label{proofoflemma2}
    The condition $\frac{\alpha_0^{(i)}}{\alpha_1^{(i)}} = \frac{\alpha_0^{(j)}}{\alpha_1^{(j)}}$ implies their reciprocals are also equal. Let $r$ be this common ratio.
    \[
        r = \frac{\alpha_1^{(i)}}{\alpha_0^{(i)}} = \frac{\alpha_1^{(j)}}{\alpha_0^{(j)}}.
    \]
    This allows us to express the probabilities as $\alpha_1^{(i)} = r\alpha_0^{(i)}$ and $\alpha_1^{(j)} = r\alpha_0^{(j)}$.

    Then
    \begin{align*}
        P_Y^{(1)}(i) &= (1 - p)W_{Y|X}(i|0) + p W_{Y|X}(i|1) \\
        &= (1 - p)\alpha_0^{(i)} + p \alpha_1^{(i)} \\
        &= (1 - p)\alpha_0^{(i)} + p (r\alpha_0^{(i)}) \\
        &= \alpha_0^{(i)}(1 - p + p r).
    \end{align*}

    By identical reasoning for symbol $j$, we have
    \[
        P_Y^{(1)}(j) = \alpha_0^{(j)}(1 - p + p r).
    \]

    For $X_2 \sim \mathrm{Ber}(q)$, substituting $q$ for $p$ yields
    \begin{align*}
        P_Y^{(2)}(i) &= \alpha_0^{(i)}(1 - q + q r), \\
        P_Y^{(2)}(j) &= \alpha_0^{(j)}(1 - q + q r).
    \end{align*}

    When symbols $i$ and $j$ are merged, the probability of observing the supersymbol $ij$ for input $X_1$ is
    \begin{align*}
        P_{Y^\prime}^{(1)}(ij) &= P_Y^{(1)}(i) + P_Y^{(1)}(j) \\
        &= (\alpha_0^{(i)} + \alpha_0^{(j)})(1 - p + p r).
    \end{align*}
    Similarly for $X_2$,
    \[
        P_{Y^\prime}^{(2)}(ij) = (\alpha_0^{(i)} + \alpha_0^{(j)})(1 - q + q r).
    \]

    The squared Hellinger distance is $H^2(P, Q) = \sum_{y} \left(\sqrt{P(y)} - \sqrt{Q(y)}\right)^2$. Because the merge only affects $i$ and $j$, the distance change is exactly the difference between the separated symbols and the merged symbol. Let $\Delta$ denote this difference:
    \[
        \Delta = H^2\left(P_Y^{(1)}, P_Y^{(2)}\right) - H^2\left(P_{Y^\prime}^{(1)}, P_{Y^\prime}^{(2)}\right).
    \]
    
    The contribution from the distinct symbols $i$ and $j$ is
    \begin{align*}
        T_{\text{orig}} &= \left(\sqrt{P_Y^{(1)}(i)} - \sqrt{P_Y^{(2)}(i)}\right)^2 \\
        &\quad + \left(\sqrt{P_Y^{(1)}(j)} - \sqrt{P_Y^{(2)}(j)}\right)^2 \\
        &= \left(\sqrt{\alpha_0^{(i)}(1 - p + p r)} - \sqrt{\alpha_0^{(i)}(1 - q + q r)}\right)^2 \\
        &\quad + \left(\sqrt{\alpha_0^{(j)}(1 - p + p r)} - \sqrt{\alpha_0^{(j)}(1 - q + q r)}\right)^2 \\
        &= (\alpha_0^{(i)} + \alpha_0^{(j)}) \left(\sqrt{1 - p + p r} - \sqrt{1 - q + q r}\right)^2.
    \end{align*}

    The contribution from the merged supersymbol $ij$ is
    \begin{align*}
        T_{\text{merged}} &= \left(\sqrt{P_{Y^\prime}^{(1)}(ij)} - \sqrt{P_{Y^\prime}^{(2)}(ij)}\right)^2 \\
        &= \Bigg(\sqrt{(\alpha_0^{(i)} + \alpha_0^{(j)})(1 - p + p r)} \\
        &\qquad - \sqrt{(\alpha_0^{(i)} + \alpha_0^{(j)})(1 - q + q r)}\Bigg)^2 \\
        &= (\alpha_0^{(i)} + \alpha_0^{(j)}) \left(\sqrt{1 - p + p r} - \sqrt{1 - q + q r}\right)^2.
    \end{align*}
Since $T_{\text{orig}} = T_{\text{merged}}$, we have $\Delta = 0$, proving the squared Hellinger distance is preserved.
\end{proof}

\section*{Proof of \Cref{corr:collinear}}
To prove the ``only if'' direction, suppose there exists $\lambda\in[0,1]$ such that $\lambda\mu_a+(1-\lambda)\mu_b=\mu_c$ for distinct $a,b,c\in\{0,1,2\}$. Further, suppose there exist $(\epsilon,\delta)$-SDHT scheme $\left(\{W_{Y_i|X_i}\}_{i=1}^n,\det\right)$ for $\cH_0$ vs. $\cH_1$ using $n$ samples and without using shared key. We will construct a $(\epsilon,\delta)$-SDHT scheme for $\{\Ber(p_0),\Ber(p_1)\}$ vs. $\{\Ber(p_2)\}$ for some $p_0,p_1,p_2\in[0,1]$ using $n$ samples and without using shared key. This contradicts Theorem~\ref{thm:impossibility} for sufficiently small $\epsilon>0$ and $\delta>0$.

Consider $R_{X|U}$ such that $R_{X|0} = \mu_a$, $R_{X|1} = \mu_b$. When $\cU=\{0,1\}$, $R_{X|U}\circ \Ber(0) = \mu_a$, $R_{X|U}\circ \Ber(1) = \mu_b$ and $R_{X|U}\circ \Ber(\lambda) = \mu_c$. Let $\phi$ be a map from $\{0,1,2\}$ to $\{0,1,\lambda\}$ that maps $a,b,c$ to $0,1,\lambda$, respectively. Then, $\{\{W_{Y_i|X_i}\circ R_{X|U}\}_{i=1}^n,\det\}$ is a $(\epsilon,\delta)$-SDHT scheme for $\{\Ber(\phi(0)),\Ber(\phi(1))\}$ vs. $\{\Ber(\phi(2))\}$ using $n$ samples and without using shared key. This contradicts Theorem~\ref{thm:impossibility} for sufficiently small $\epsilon>0$ and $\delta>0$.

To prove the ``if'' direction, we will show that \Cref{prop:lp} implies a SDHT scheme using $n$ samples that incurs $n$ bits of communication without using shared secret key when $\lambda,a,b,c$ satisfying the condition in the corollary do not exist, 

For this, we will construct $W_{Y|X}$ such that $W_{Y|X}\circ\mu_0=W_{Y|X}\circ\mu_1$, but $W_{Y|X}\circ\mu_2\neq W_{Y|X}\circ\mu_0$.
Interpret $\mu_0,\mu_1$ and $\mu_2$ as $|\cX|$-dimensional vectors; define $\bm{u}_1=\mu_0-\mu_1$ and $\bm{u}_2=\mu_0-\mu_2$.
By non-collinearity, $\bm{u}_1$ and $\bm{u}_2$ are linearly independent.
Let $\bm{v}$ be a unit vector such that $\langle \bm{v},\bm{u}_1\rangle = 0$ and $\langle \bm{v},\bm{u}_2\rangle > 0$.
When $\bm{1}$ is the all-ones vector,
there exists $\beta>0$ such that each coordinate of $\bm{v}'=\alpha\cdot\bm{v} + \frac{1}{2}\cdot \bm{1}$ belongs to $[0,1]$.
Define $W_{Y|X}$ such that $W_{Y|X}(x)=\mathrm{Bernoulli}(\bm{v}'_x)$ for $x\in\cX$.
Clearly, for any $x\in\cX$, $\Pr[W_{Y|X}(x)=1] = \bm{v}'_x$, where $\bm{v}'_x\in[0,1]$ is the $x$-th coordinate of $\bm{v}'$. But then, for $b\in\{1,2\}$,
\begin{align*}
    &W_{Y|X}\circ \mu_0(1) - W_{Y|X}\circ \mu_b(1) \\
    &\quad = \langle \bm{v}',\bm{u}_b\rangle = \alpha\cdot\langle \bm{v},\bm{u}_b\rangle + \frac{1}{2}\cdot\langle \bm{1},\bm{u}_b\rangle.
\end{align*}
Clearly, $\langle \bm{1},\bm{u}_b\rangle = 0$ for $b\in\{1,2\}$ since $\mu_0,\mu_1$ and $\mu_2$ are distributions. Moreover, $\langle \bm{v},\bm{u}_1\rangle = 0$ and $\langle \bm{v},\bm{u}_2\rangle > 0$. We conclude that the channel $W_{Y|X}$ maps $\mu_0$ and $\mu_1$ to the same distribution, but maps $\mu_2$ to a different distribution.

\section*{Proof of Theorem~\ref{thm:achievability}}
By \Cref{corr:collinear}, it suffices to construct perfectly private distributed hypothesis test for $\cH_0=\{\mu_0,\mu_1\}$ and $\cH_1=\{\mu_2\}$ such that $\mu_0,\mu_1$ and $\mu_2$ are collinear. This is shown in two steps. In \Cref{clm:bernoulli-achievability}, a perfectly private distributed hypothesis test is constructed for $\cH_0$ vs. $\cH_1$ using $1$ bit shared key when $\mu_0,\mu_1$ and $\mu_2$ are Bernoulli distributions such that the latter is distinct from the former two. Then we argue that, when $\mu_0,\mu_1$ and $\mu_2$ are collinear distributions such that $\mu_2$ is distinct from both $\mu_0$ and $\mu_1$, there exists a channel $W_{Y|X}$ that maps $\mu_0,\mu_1$ and $\mu_2$ to Bernoulli distributions where the latter is distinct from the former two. By composing the channel $W_{Y|X}$ with the perfectly private distributed hypothesis test for Bernoulli distributions, we obtain the required perfectly private distributed hypothesis test for $\cH_0$ vs. $\cH_1$ using $1$ bit shared key.

The latter step is straightforward: let $\cY=\{0,1\}$. Let $W_{Y|X}$ be the deterministic channel which maps each $x\in \cX$ to $1$ if $\mu_0(x) > \mu_1(x)$ and maps $x$ to $0$ otherwise. Since $\mu_0$ and $\mu_2$ are distinct, $W_{Y|X}$ maps $\mu_0$ and $\mu_2$ to different Bernoulli distributions. Moreover, since $\mu_0,\mu_1$ and $\mu_2$ are collinear, but $\mu_1$ and $\mu_2$ are distinct, $W_{Y|X}$ maps $\mu_1$ and $\mu_2$ to distinct Bernoulli distributions.

\begin{claim}\label{clm:bernoulli-achievability}
    Let $\cH_0=\{\Ber(p_0),\Ber(p_1)\}$ and $\cH_1=\{\Ber(p_2)\}$ where $p_2\notin\{p_0,p_1\}$. Then, there exists a perfectly private distributed hypothesis test using $1$ bit shared key for $\cH_0$ vs. $\cH_1$ using $n$ samples that is $(1-e^{-\Omega(n)})$-correct and incurs $n$ bits of communication.
\end{claim}
\begin{proof}
    First, consider the case where, for some $q\notin\{p,1-p\}$, $\mu_0=\Ber(p)$, $\mu_1=\Ber(1-p)$ and $\mu_2=\Ber(q)$. Let the shared key $K$ be a uniformly random bit. For each $i$, let $W_{Y_i|X_i,K}=X_i\oplus K$.
    Let $\det$ be a hypothesis test for $\cH'_0=\{\Ber(p),\Ber(1-p)\}$ vs. $\cH'_1=\{\Ber(q),\Ber(1-q)\}$ using $n$ samples with correctness error at most $\epsilon$. We will argue that $\left(\{W_{Y_i|X_i,K}\}_{i=1}^n,\det\right)$ is a perfectly private distributed hypothesis test for $\cH_0$ vs. $\cH_1$ with correctness error at most $\epsilon$.
    For $i\in\{0,1,2\}$, let
    \[ P^{i}_{X^nKY^n} = \Ber(p_i)^{\otimes n}\cdot P_K\cdot \prod_{i=1}^n W_{Y_i|X_i,K} \]
    Clearly, when the test distribution is $\Ber(p_i)$, $(Y_1,\ldots,Y_n)$ are distributed according to $P^i_{Y^n}$ for $i\in\{0,1,2\}$.
    Observe,  for $a,b\in\{0,1\}$, conditioned on $K=b$, $P^{a}_{Y^n} = \Ber(p_{a\oplus b})^{\otimes n}$.
    Furthermore, $P^{2}_{Y^n} = \Ber(b\cdot(1-q)+(1-b)\cdot q)^{\otimes n}$ conditioned on $K=b$.
    By the first observation, $P^0_{Y^n}$ and $P^1_{Y^n}$ are identical distributions---they are distributed according to $\Ber(p)^{\otimes n}$ with probability $1/2$ and according to $\Ber(1-p)^{\otimes n}$ with probability $1/2$---ensuring perfect privacy.
    A similar statement holds for $P^2_{Y^n}$, but with $q$ in place of $p$ by the second observation.
    For any realization of the shared key $K$, $(Y_1,\ldots,Y_n)$ are distributed i.i.d. according a distribution in $\cH'_0$ (resp., $\cH'_1$) when the test distribution is in $\cH_0$ (resp., $\cH_1$).
    Thus, the correctness error of the test reduces to the correctness error of $\det$ for $\cH'_0$ vs. $\cH'_1$, which is at most $\epsilon$. An optimal test $\det$ will achieve correctness error $\epsilon = e^{-\Omega(n)}$. The communication incurred by the test $\det$ is exactly $n$ bits.

    In the rest of the proof, we will prove the following:
    \begin{claim}
        There exists $W_{Y|X}$ and $p,q$ such that $q\notin\{p,1-p\}$, $W_{Y|X}\cdot \Ber(p_0)=\Ber(p)$, $W_{Y|X}\cdot \Ber(p_1)=\Ber(1-p)$ and $W_{Y|X}\cdot \Ber(p_2)=\Ber(q)$ when $p_2\notin\{p_0,p_1\}$ and $p_0\neq p_1$.
    \end{claim}
    Note, when $p_0=p_1$, the distributions $\mu_0$ and $\mu_1$ are identical, and thus, the test is trivial.
    Otherwise, composing the channel $W_{Y|X}$ with the perfectly private distributed hypothesis test for Bernoulli distributions, we obtain the required perfectly private distributed hypothesis test for $\cH_0$ vs. $\cH_1$ using $1$ bit shared key. This concludes the proof of Theorem~\ref{thm:achievability}.
    
    We now prove the above claim.
    To find $W_{Y|X}$, let $W_{Y|X}(0|0)=\alpha$ and $W_{Y|X}(0|1)=\beta$ for some $\alpha,\beta\in[0,1]$.
    Define the function $f:t\mapsto W_{Y|X}\circ\Ber(t)(0)$ for $t\in[0,1]$. 
    It is easy to see that $f(t) = \alpha\cdot t + \beta\cdot (1-t)$; alternatively, $f(t) = m\cdot t + k$ where $m=\alpha-\beta$ and $k=\beta$.

    We have, $f(p_0)+f(p_1)=p+1-p=1$.
    Hence, 
    $$m\cdot (p_0+p_1) + 2k = 1 \implies k = \frac{1}{2}-\frac{m\cdot (p_0+p_1)}{2}.$$
    Forcing the constraints $\alpha= m+k\in[0,1]$ and $\beta=k\in[0,1]$ and using the above expression for $k$, we have the following constraints on $m$:
    \begin{align*}
        -\frac{1}{2-p_0-p_1} \le m \le \frac{1}{2-p_0-p_1}\\
        -\frac{1}{p_0+p_1} \le m \le \frac{1}{p_0+p_1}
    \end{align*}
    Since $0<p_0+p_1<2$, we can choose $0\neq m$ that satisfies the above inequalities.
    With this choice of $m$, we have $f(p_0)+f(p_1)=1$, $f(p_2)\notin \{f(p_0),f(p_1)\}$.
    The latter follows from the fact that $f$ is injective since $m\neq 0$, and $p_2\notin\{p_0,p_1\}$. 
    Since $\alpha$ and $\beta$ associated with $f$ are in $[0,1]$, $W_{Y|X}$ is a valid channel. This concludes the proof of the claim.
\end{proof}

\section*{Missing Details in Remark~\ref{remark}} 
\begin{lemma}
$\dtv(P,Q) \ge \delta$ if and only if there exists a deterministic test $D$ that uses sample $X$ and outputs $D(X) \in \{0,1\}$ such that
\[
\bP_P[D(X) = 0] + \bP_Q[D(X) = 1] \ge 1+\delta.
\]
\end{lemma}
\begin{proof}
    For any deterministic test $D$, we have
    \begin{align*}
        &\bP_P[D(X) = 0] + \bP_Q[D(X) = 1] \\
        &= 1 - \bP_P[D(X) = 1] + \bP_Q[D(X) = 1]\\
        &= 1 + Q(A) - P(A).
    \end{align*}
    where  $A = \{x : D(x) = 1\}$.
Suppose $\dtv(P,Q) < \delta$. Then, since $\dtv(P,Q) = \max_{A\subseteq \cX} Q(A) - P(A)$, we have
    \begin{align*}
        \bP_P[D(X) = 0] + \bP_Q[D(X) = 1] & \le 1 + \dtv(P,Q) \\
        &< 1 + \delta.
    \end{align*}
Next, suppose that $\dtv(P,Q) \ge \delta$ and consider the test $D^*(x) = \indic_{\{x \in A^*\}}$ where $\dtv(P,Q) = Q(A^*)-P(A^*)$. Then, from the calculation above, we have
\begin{align*}
    \bP_P[D^*(X) = 0] + \bP_Q[D^*(X) = 1] &= 1 + \dtv(P,Q) \\
    &\ge 1+\delta.
\end{align*}
\end{proof}
\fi
\end{document}